\documentclass[12pt,english,a4paper]{article}

\usepackage[english]{babel}

\usepackage{amssymb}

\usepackage[T1]{fontenc}

\usepackage{enumerate}

\usepackage{amscd}

\usepackage{amsthm}

\usepackage{amsmath}

\usepackage{array}

\usepackage{delarray}

\usepackage{tikz}

\def\ra{\longrightarrow}

\def\emptyset{\varnothing}

\def\<{\langle}

\def\>{\rangle}

\def\Z{\mathbb{Z}}

\def\Q{\mathbb{Q}}

\def\F{\mathbb{F}}

\def\isom{\cong}

\def\S{\mathcal{S}}

\newtheorem{theo}{Theorem}[section]

\newtheorem{cor}[theo]{Corollary}

\newtheorem{prop}[theo]{Proposition}

\theoremstyle{definition}

\newtheorem{Def}[theo]{Definition}

\newtheorem{ex}[theo]{Example}

\newtheorem{remark}[theo]{Remark}
\begin{document}


\title{Finite $p$-groups, entropy vectors and the Ingleton inequality for nilpotent groups}
\author{Pirita Paajanen\footnote{The author is supported by an Academy of Finland Postdoctoral Fellowship}\\ Department of Mathematics and Statistics\\ PO Box 68\\00014 University of Helsinki\\ Finland\\ pirita.paajanen@helsinki.fi}
\maketitle

\begin{abstract}
In this paper we study the capacity/entropy region of finite, directed, acyclic, multiple-sources, multiple-sinks network by means of group theory and entropy vectors coming from groups. 
There is a one-to-one correspondence between the entropy vector of a collection of $n$ random variables and a certain group-characterizable vector obtained from a finite group and $n$ of its subgroups.  We are looking at nilpotent group characterizable entropy vectors and show that they are all also Abelian group characterizable, and hence they satisfy the Ingleton inequality. It is known that not all entropic vectors can be obtained from Abelian groups, so our result implies that in order to get more exotic entropic vectors, one has to go at least to soluble groups or larger nilpotency classes. The result also implies that Ingleton inequality is satisfied by nilpotent groups of bounded class, depending on the order of the group. 
\end{abstract}


\section{Introduction}
The coding theoretic background to this paper is the problem of determining the capacity region  of a finite, directed, acyclic multiple-sources multiple-sinks network. It is a notoriously difficult question to determine the precise coding capacity of such a network, and to make any progress we start by giving bounds for this capacity. This capacity region is essentially the same as characterizing the entropy region. 

Let $\mathcal{N}=\{ 1,2,\dots,n\}$ and let $X_1,X_2,\dots, X_n$ be $n$ jointly distributed discrete random variables. For any nonempty subset $\S\subseteq \mathcal{N}$, we write $X_\S=\{X_i:i\in \S\}$, with joint entropy $h_\S=H(X_i:i\in\S)$. We call the ordered $(2^n-1)$-tuple $$\mathfrak{h}=(h_\S)_{\emptyset\neq\S\subseteq \mathcal{N}}$$ an entropy vector.
The region $\Gamma^*_n$ is defined to consist of all entropic vectors, while the region $\Gamma_n$ is defined to be the subset of $2^n-1$-dimensional real space bounded by all standard information inequalities, that is, all Shannon-type inequalities. Now, it is standard knowledge, that $\Gamma_n^*$ is contained in $\Gamma_n$. To bound the region $\Gamma_n^*$ we might need also so-called non-Shannon-type inequalities. For $n=2,3$ the regions are known to be bound only by Shannon-type inequalities, but in the case $n=4$ we genuinely need also non-Shannon-type inequalities. There are a number of different inequalities found by several people, and these non-Shannon-type inequalities give a region between $\Gamma_4^*$ and $\Gamma_4$.  One of these inequalities is the Ingleton inequality given by Mat\'{u}\v{s} \cite{Matus}, since Mat\'{u}\v{s} and Studen\v{y} \cite{Matus-Studeny} showed that every vector in $\Gamma_4$ satisfying the Ingleton inequality is entropic, that is, it is in $\Gamma_4^*$. This gives us an inner bound for $\Gamma_4^*$.

An important and somewhat surprising tool to characterize the entropy region comes from group theory. It is well-known that there is a one-to-one correspondence between the entropy vector of a collection of $n$ random variables and a certain group-characterizable vector obtained from a finite group and $n$ of its subgroups \cite{Chan-Yeung}, which we will describe next. A more detailed exposition can be found in the well-written introduction of Boston and Nan in \cite{Boston}.

Let $G$ be a finite group, and let $G_1,\dots, G_n$ be a set of $n$ of its subgroups. Let $\mathcal{N}=\{1,2,\dots,n\}$. For all nonempty subsets of $\mathcal{S} \subseteq\mathcal{N}$, we define another subgroup, that is the intersection $G_{\mathcal{S}}=\cap_{i\in \mathcal{S}} G_i$. We write $g_\mathcal{S}=\log|G:G_\mathcal{S}|$. This vector turns out to be entropic. However, not every entropic vector can be characterized in this way.

\begin{Def}
Let $\mathfrak{h}$ be the entropic vector associated with $n$ random variables $X_1,\dots,X_n$. We call $\mathfrak{h}$ {\bf group characterizable} if there is a finite group $G$ and $n$ subgroups $G_1,\dots,G_n$ of $G$ such that $\mathfrak{h}=\mathfrak{g}=(g_\mathcal{S})_{\emptyset\neq\mathcal{S}\subseteq\mathcal{N}}$.
Similarly, we call $\mathfrak{h}$ {\bf Abelian group characterizable}, if the finite group $G$ is Abelian.
\end{Def}

We denote the region of all group-characterizable vectors by $\gamma_n$ and note that this is contained in $\Gamma_n^*$. Furthermore, this region $\gamma_n$ is nearly what we are interested in, since the conic closure of $\gamma_n$ is equal to the closure of $\Gamma_n^*$. Hence, studying group characterizable vectors is enough to study the region of $\Gamma_n^*$. In particular, for $n=4$, it is enough to study the gap between those vectors satisfying the Ingleton inequality and the true region of $\Gamma_4^*$ by using group characterizable vectors. Then Ingleton inequality in the context of groups can be written as \[|G_1||G_2||G_{34}||G_{123}||G_{124}|\geq |G_{12}||G_{13}||G_{14}||G_{23}||G_{24}|.\]

It is known that if an entropic vector is Abelian group characterizable, then the Ingleton inequality holds. In this paper we are concentrating on identifying Abelian group characterizable vectors. In particular we show that entropic vectors coming from certain nilpotent group are Abelian group characterizable, and satisfy the Ingleton inequality. Nilpotent groups is the simplest class of groups that are not Abelian, and our result implies that network codes coming from them are no better than linear codes.

In \cite{Mao-Hassibi} Mao and Hassibi study which groups violate the Ingleton inequality, and identify the symmetric group on 5 letters $S_5$, as the smallest such group, and more generally show that $\mathrm{PGL}_2(p)$ all violate the Ingleton inquality. The work of Boston and Nan \cite{Boston} finds more examples of groups that violate the Ingleton inequality, including soluble groups. So far the case of nilpotent groups is open. Stancu and Oggier claim in \cite{Stancu-Oggier} that nilpotent groups satisfy the Ingleton inequality, however there is a mistake in the proof of Theorem 1, also known to the authors, namely, an intersection of quotients in not necessarily equal to a quotient of intersection, that brings the whole argument down. 
In the current paper we do get partial results  and we show that some nilpotent groups satisfy the inequality. There is a technical condition in the methodology of this paper, that prevents a general proof. In general we do expect the results to hold for all nilpotent groups, but this would require a whole new methodology.

The main results of this paper are:

\begin{enumerate}
\item  The set of $p$-group characterizable vectors for $p$-groups of class $c<p$
is contained in the set of Abelian group characterizable vectors provided class $c<p$.
\item  The set of nilpotent group characterizable vectors is Abelian group characterizable, if their class is smaller than the smallest prime dividing the order of the group. In particular, all class two nilpotent groups of odd order produce Abelian characterizable vectors.
\item $p$-groups of class $c< p $ satisfy the Ingleton inequality.
\item Finite nilpotent groups satisfy the Ingleton inequality, provided their class is smaller than the smallest prime dividing the order of the group. In particular, all class two nilpotent groups of odd order satisfy the Ingleton inequality.
 \end{enumerate}
 
 \section{$p$-groups}
 
 We begin with a number of group theoretic definitions.
 
 \begin{Def}
A finite group $G$ whose order is $p^n$, for some prime $p$ and $n\geq 1$ is called a {\bf $p$-group}. 
\end{Def}

\begin{Def}
Let $x,y\in G$, then we call $[x,y]=x^{-1}y^{-1}xy$ the {\bf commutator} of $x$ and $y$. The {\bf commutator subgroup} is the subgroup generated by all commutators, $[G,G]=\< [x,y]: x,y\in G \>$.\end{Def}

\begin{Def}\label{lower}
The {\bf lower central series} for a group $G$ is a sequence of subgroups $\gamma_i(G)$ defined as follows. The first term is {$\gamma_1(G)=G$} and inductively $\gamma_{i+1}(G)=[\gamma_i(G),G]$ for $i\geq 1$. The group $G$ is called {\bf nilpotent of class $c$} if $\gamma_{c+1}(G)=1$, where $c$ is the least when this happens. \end{Def}

All this background material and the following proposition can be found in any standard book on group theory, see \cite{Alperin-Bell} for example.

\begin{prop}
All finite $p$-groups are nilpotent.
\end{prop}

The results of the current paper are based on Lazard correspondence, that allows us to define a Lie ring structure from the group $G$. First we define Lie ring. 

\begin{Def}
 A {\bf Lie ring} $L$ is an additive Abelian group $L$ together with an operation $[\cdot,\cdot]:L\times L\ra L$  satisfying
 the following identities
 \begin{enumerate}
 \item $[x,x]=0$
 \item $[y,x]=-[x,y]$
 \item $[x,[y,z]]+[y,[z,x]]+[z,[x,y]]=0$
 \end{enumerate} for all $x,y,z\in L$. The Property 3. above is called the Jacobi identity. We call $[,]$ the {\bf Lie bracket} or the Lie commutator.
 \end{Def}

\begin{ex}\label{Heisenberg} Let $\F_p$ be the field with $p$ elements.
The Heisenberg Lie ring consists of the set of matrices
\[\mathcal{H}=
\left\{\left(
\begin{array}{ccc}
 0 & a_{12} & a_{13}\\
0 &0&a_{23}\\
0&0&0
\end{array}
\right): a_{12},a_{13}, a_{23}\in \F_p\right\}
\] with matrix addition and the Lie bracket $[A,B]=AB-BA$ for $A,B\in \mathcal{H}$.
\end{ex}

Next we construct a Lie ring structure from a finite $p$-group. This construction is known as the Lazard correspondence, and mathematicians working with $p$-groups have been familiar with it since the 1950s.  We first illustrate the construction by an example of a group of $3^3=27$ elements.
  
 \begin{ex} We look at an example over $\F_p$, the field with 3 elements.
 Let $G$ be a group consisting of the $3\times 3$ matrices
\[\mathcal{G}=
\left\{\left(
\begin{array}{ccc}
 1 & b_{12} & b_{13}\\
0 &1&b_{23}\\
0&0&1
\end{array}
\right): b_{12},b_{13},b_{23}\in \F_3\right\}
\] where the group operation is the normal matrix multiplication. 

Let $g,h\in G$. We now define a new addition on this set, by putting $g\oplus h:=gh[h,g]^2$, which gives
\begin{align*}&\left(
\begin{array}{ccc}
 1 & b_{12} & b_{13}\\
0 &1&b_{23}\\
0&0&1
\end{array}
\right)\oplus
\left(
\begin{array}{ccc}
 1 & c_{12} & c_{13}\\
0 &1&c_{23}\\
0&0&1
\end{array}
\right)
=\\&\left(
\begin{array}{ccc}
 1 & b_{12}+c_{12} & b_{13}+c_{13}+2b_{12}c_{23}+2c_{12}b_{23}\\
0 &1&b_{23}+c_{23}\\
0&0&1
\end{array}
\right)\end{align*}
using the fact that the entries of the matrices are considered $\mod 3$. This operation is commutative, and on the $(1,2), (2,3)$-entries it is just addition of entries, but the addition on $(1,3)$-entry is twisted, so that this new addition is not the same addition as normal addition of matrices. 
However, this new addition gives an Abelian group structure. It is easy to see that the group commutator corresponds to the Lie commutator defined in the previous example. This Lie ring is isomorphic to the  Heisenberg Lie ring of Example \ref{Heisenberg} in the case that $p=3$.
\end{ex}

 More generally for groups of nilpotency class two,
 the addition in the Lie algebra setting will be defined using the group operation as follows:
$$x+y:=xy[y,x]^{\frac{1}{2}},$$
where $[,]$ is the group commutator and for $z\in G$ the exponentiation means $z^{\frac{1}{2}}:=z^{(p^k+1)/2}$ and $p^k$ is the order of $z$. The Lie bracket will be simply the commutator operation on $G$. 

For a general $p$-group the construction goes via the Hausdorff formula. the Hausdorff formula gives an expression for the formal power series $$\Phi(X,Y)=\log(\exp(X)\cdot \exp(Y))\in \Q((X,Y))$$ in non-commuting indeterminates $X,Y$. The logarithm and exponential are defined $$\log(1+X)=\sum_{n=1}^\infty (-1)^{n-1}X^n/n\ \mathrm{and}\ \exp(X)=\sum_{n=0}^\infty X^n/n!$$ as the usual formal power series.
The Hausdorff formula is an infinite series, and it starts $$X+Y+\frac{1}{2}[X,Y]+\frac{1}{12}[X,[X,Y]]  +\frac{1}{12}[Y,[Y,X]]+\dots $$ see e.g. Serre \cite{Serre}. Later terms contain Lie brackets/commutators of increasing length.
When the group/Lie ring is nilpotent, commutators of only certain finite length will survive, and so for nilpotent groups the Hausdorff formula is truncated into a finite formula. Further, we require the nilpotency class to be $c < p$ for the correspondence to hold, so that the denominators of the coefficients of the series do not interfere with the algebraic structure. See Khukhro \cite{Khukhro} for details. Here we just formulate the theorem in the form that it is used.

\begin{theo}[Lazard's correspondence]
The Hausdorff formula and its inverse set up a correspondence between finite $p$-groups of nilpotency class less than $p$, and nilpotent Lie rings of class less than $p$ whose additive group is a finite $p$-group. The correspondence preserves the orders and nilpotency classes of the objects; these extend to subgroups and Lie subrings.\end{theo}
   
 \begin{theo}[Main Theorem]
 Let $G$ be a finite $p$-group of nilpotency class $c$. If $c<p$, then the entropic vector $\mathfrak{g}$ associated with $G$ is Abelian group characterizable for any $n$, i.e.,  there is an Abelian group $A$ and subgroups $A_1,\dots,A_n$ such that  for all nonempty $\emptyset\neq\mathcal{S}\subseteq\mathcal{N}$ we have 
 $\log|G:G_\mathcal{S}|=\log|A:A_\mathcal{S}|$ where
 $G_{\mathcal{S}}=\cap_{i\in \mathcal{S}} G_i$. 
 \end{theo}

 \begin{proof}
 We will show that there is a bijection from $G$ to an Abelian group  $A$, which is identity on the elements, and sends a subgroup $G_i\leq G$ to a subgroup $A_i\leq A$. Moreover it is constructed in such a way that it respects the inclusion and intersections of the subgroups.

 Let $S$ be the underlying set of $G$,  and let us denote group structure on $S$ by $G=(S,*)$ and the Lie ring structure by $L=(S,+,[,])$. 
The Lazard correspondence sets up a bijection 
 \begin{align*}
\Psi: G&\ra L\\
\end{align*} 
 which is identity on the elements, sends subgroups to subalgebras, normal subgroups to ideals, and preserves indeces of subgroups, and respects intersections, since it is identity on elements. The Lie bracket corresponds to the group commutator.
 
 Further we can define an Abelian group $A=(S,+)$ from the Lie ring by simply forgetting the bracket operation. All  Lie subrings are closed as Abelian groups, and hence each Lie subring $L_i$ corresponds to a subgroup $A_i\leq A$, so we have a map \begin{align*} \Phi:L&\ra A.\\ 
\end{align*} Now the map $\Phi\circ\Psi:G\ra A$ gives us the desired correspondence between subgroups $G_i\leq G$ and Abelian subgroups $A_i\in A$ and hence $G$ is uniformly Abelian representable. 
 \end{proof}

\begin{remark}
There  may be more subgroups in an Abelian group $A$ than in $G$, so the map may not give a one-to-one correspondence of subgroups. 
\end{remark} 
 
 \begin{ex}
 Let us consider the example of a group with presentation $$G=\<x,y: x^{p^2}=y^p=1, [x,y]=x^p \>$$ for $p\geq 3$. The group $G$ is nilpotent of class 2, so we can apply the Lazard correspondence, giving us a Lie ring with the same presentation.  The group $G$ has $p+1$ subgroups of index $p$, and $p+1$ subgroups of index $p^2$, and since the orders of the subgroups are so small, all subgroups are Abelian. Indeed, the Abelian group $\Z/p^2\Z\oplus \Z/p\Z$ has identical subgroup structure. In the Figure \ref{fig: Abelian group example}, we give the case $p=3$.

 \begin{figure}[htbp]
\begin{center}
\begin{tikzpicture}
\node (G) at (3,3) {$G$};
\node (A) at (2,2) {$G_1$}
edge (G);
\node (B) at (3,2) {$G_2$}
edge (G);
\node (C) at (4,2) {$G_3$}
edge (G);
\node (D) at (5,2) {$G_4$}
edge (G);
\node (a) at (0,1) {$H_1$}
edge (A);
\node (b) at (1,1) {$H_2$}
edge (A);
\node (c) at (2,1) {$H_3$}
edge (A);
\node (d) at (3,1) {$H_4$}
edge (A)
edge (B)
edge (C)
edge (D);
\node at (3,0) {$\{1\}$}
edge (a)
edge (b)
edge (c)
edge (d);
\end{tikzpicture}
\end{center}
\caption{Subgroups of $\Z/9\Z\oplus \Z/3\Z$.}\label{fig: Abelian group example}
\end{figure} 
\end{ex}

 Next we define the concept of a subgroup lattice. A pictorial illustration is given in Figure \ref{fig: Abelian group example}. This will be used in applications to Ingleton inequality.
 Let $G$ be a finite group and let $G_1,\dots,G_k$ be the set of its subgroups. The subgroups can be ordered by inclusion and they form a partially ordered set. This set is also a lattice, and furthermore the subgroup lattice can be expressed as a graph, whose vertices are the subgroups $G_1=G, G_2,\dots,G_k=1$ and there is an edge between $G_i$ and $G_j$ if $G_i\leq G_j$. Intersections and joins of subgroups are easily identified in the subgroup lattice.
 
 \begin{cor}
 The subgroup lattice of a $p$-group $G$ or order $p^n$ of class $c<p$ can be embedded in the subgroup lattice of the Abelian group $A$.
 \end{cor}
 
 \begin{proof}
 Follows directly from the properties of the map $\Phi\circ\Psi$.
 \end{proof}
 
 \begin{cor}
  A $p$-group $G$ of class $c<p$ satisfies the Ingleton inequality.
 \end{cor}
 
 \begin{proof}
 The subgroup lattice of $G$ can be embedded in the subgroup lattice of $A$. So given any 4 subgroups $G_1,\dots, G_4$ of $G$ and their intersections, there is an order preserving map to subgroups $A_1,\dots,A_4$ of $A$ and their intersections. Since $A$ is Abelian, the subgroups $A_1,\dots,A_4$ and their intersections satisfy the Ingleton inequality. 
 \end{proof}

 \section{General nilpotent groups}
 
 A finite group $G$ is nilpotent if it has a lower central series as in the Definition \ref{lower}. In practical problems the following proposition allows us to reduce questions about nilpotent groups into $p$-groups. 
 
 \begin{Def}
 Let $G$ be a finite group of order $|G|=mp^k$ where $p\nmid m$, then a subgroup $P$ of order $p^k$ is called a {\bf Sylow-$p$-subgroup} of $G$.
 \end{Def}
 
\begin{prop}
A finite nilpotent group is a direct product of its Sylow-$p$-subgroups. That is, $G\isom P_1\times P_2\times \dots \times P_k$, where $|P_i|=p_i^{k_i}$ and the $p_i$ are all distinct primes.
\end{prop}

\begin{prop}\label{direct product}
A nilpotent group satisfies the Ingleton inequality if all of its Sylow-$p$-subgroups satisfy the Ingleton inequality.
\end{prop}

\begin{proof}
Suppose all the Sylow-$p$-subgroups of $G$ satisfy the Ingleton inequality. 

Any subgroup  $H\leq G$ can be written as  $H\isom (H_{p_1},\dots, H_{p_k})$. This implies that $|H|=|H_{p_1}||H_{p_2}|\dots |H_{p_k}|$.
We note that if $H,K\leq G$ then $H\cap K=(H_{p_1}\cap K_{p_1},\dots, H_{p_k}\cap K_{p_k})$.
So let $H_1,H_2,H_3,H_4\leq G$. Since the Sylow-$p_i$-subgroups satisfy the Ingleton inequality we have 
\begin{align*}&|H_{1,p_i}| |H_{2,p_i}| |H_{34,p_i}| |H_{123,p_i}| |H_{124,p_i}|\geq \\ &|H_{12,p_i}| |H_{13,p_i}| |H_{14,p_i}| |H_{23,p_i}| |H_{24,p_i}|\end{align*} for all divisors $p_i\mid |G|$.
Now using the fact that $|H_{j}|=|H_{j,p_1}||H_{j,p_2}|\dots |H_{j,p_k}|$ we get
\begin{align*}&|H_{1}| |H_{2}| |H_{34}| |H_{123}| |H_{124}|\geq \\ & |H_{12}| |H_{13}| |H_{14}| |H_{23}| |H_{24}|,\end{align*} as wanted.
\end{proof}

\begin{cor} Let $G$ be a nilpotent group of class $c$.
If $c$ is strictly smaller than any prime factor of the order of $G$, then $G$ satisfies the Ingleton inequality.
\end{cor}

\begin{proof}
Since the class of $G$ is $c<p_i$ for all Sylow-$p_i$-subgroups of $G$, each of these Sylow-$p_i$-subgroups of $G$ satisfy the Ingleton inequality and hence so does $G$ by Proposition \ref{direct product}.
\end{proof}

As special case we state the following theorem.

\begin{cor}
All class 2 nilpotent groups with odd order are Abelian representable and satisfy the Ingleton inequality.
\end{cor}

\begin{proof}
If $G$ has odd order, then $p_1>2$ and the above corollary holds.
\end{proof}

The last application is to Abelian representability of finite nilpotent groups, and we recall Proposition 20 in \cite{Oggier2}.

\begin{prop}
A finite nilpotent group is Abelian representable if all of its Sylow-p-subgroups are Abelian representable.
\end{prop}

\begin{cor}
A finite nilpotent group is Abelian representable, if the smallest prime dividing the order of the group is bigger than the nilpotency class.
\end{cor}



\section{Conclusion}

We have shown that nilpotent groups, with a condition on their nilpotency class are Abelian representable. Soluble groups are not  Abelian group representable. Already the smallest non-Abelian group $S_3$ is not Abelian representable as noted in \cite{Oggier1}. 
Hence soluble groups might be the first class of groups to bring new entropy vectors. Boston and Nan have found examples of soluble groups that violate the Ingleton inequality \cite{Boston}, which proves that they might be the first groups to provide interesting network codes. The other classes of groups that have been identified to violate the Ingleton inequality by Mao and Hassibi include $\mathrm{PGL}_2(p)$.  These groups are not soluble since they contain $\mathrm{PSL}_2(p)$ as normal subgroup, and the $\mathrm{PSL}_2(p)$  are simple. Hence, we would speculate that simple groups would  be a better test ground for constructing interesting network codes and entropy vectors.


\end{document}